\documentclass[journal]{IEEEtran}
\IEEEoverridecommandlockouts
\usepackage{graphicx}
\usepackage{amsmath,amsthm,amsfonts,amssymb}
\usepackage{cite}
\usepackage{bm}
\usepackage{bbm}
\usepackage{url}
\usepackage{array}
\usepackage{color,soul}
\usepackage{multirow}
\usepackage{booktabs}
\usepackage[table,xcdraw]{xcolor}
\usepackage{enumitem}
\theoremstyle{plain}
\newtheorem{thm}{Theorem}

\newtheorem{prop}[thm]{Proposition}

\newtheorem*{rem}{Remark}

\usepackage{subcaption}

\usepackage[english]{babel}
\usepackage{algorithm}
\usepackage[noend]{algpseudocode}

\allowdisplaybreaks[4]

\begin{document}
\title{Connecting the Unconnectable through Feedback}
\author{Yimeng Li and Yulin Shao
	\thanks{The authors are with the State Key Laboratory of Internet of Things for Smart City and the Department of Electrical and Computer Engineering, University of Macau, Macau S.A.R. Y. Shao is also with the Department of Electrical and Electronic Engineering, Imperial College London (emails:\{MC35385,ylshao\}@um.edu.mo).}
}
	
\maketitle
\begin{abstract}
Reliable uplink connectivity remains a persistent challenge for IoT devices, particularly those at the cell edge, due to their limited transmit power and single-antenna configurations. This paper introduces a novel framework aimed at connecting the unconnectable, leveraging real-time feedback from access points (APs) to enhance uplink coverage without increasing the energy consumption of IoT devices. At the core of this approach are feedback channel codes, which enable IoT devices to dynamically adapt their transmission strategies based on AP decoding feedback, thereby reducing the critical uplink SNR required for successful communication. Analytical models are developed to quantify the coverage probability and the number of connectable APs, providing a comprehensive understanding of the system's performance. Numerical results validate the proposed method, demonstrating substantial improvements in coverage range and connectivity, particularly for devices at the cell edge, with up to a $51\%$ boost in connectable APs.
Our approach offers a robust and energy-efficient solution to overcoming uplink coverage limitations, enabling IoT networks to connect devices in challenging environments.
\end{abstract}
	

\section{Introduction}
In the rapidly expanding realm of the Internet of Things (IoT), reliable connectivity for devices at coverage edges remains a persistent challenge \cite{kanj2020tutorial,parkvall2017nr,shao2024theory}. Fig.~\ref{fig:Introduction}(a) depicts a frequently encounter situation where an IoT device at the cell edge is still able to ``hear'' downlink transmissions from its access point (AP), yet the AP struggles to detect the device's uplink signals. This disparity arises from the limited transmit power and single-antenna design typical of many IoT devices, resulting in a pronounced imbalance between downlink and uplink performance. Overcoming this ``unconnectable'' gap is essential for unlocking the full potential of IoT applications -- particularly those spanning smart cities, healthcare, and industrial automation -- that demand wide and deep coverage.

\begin{figure}
	\centering
	\includegraphics[width=0.9\columnwidth]{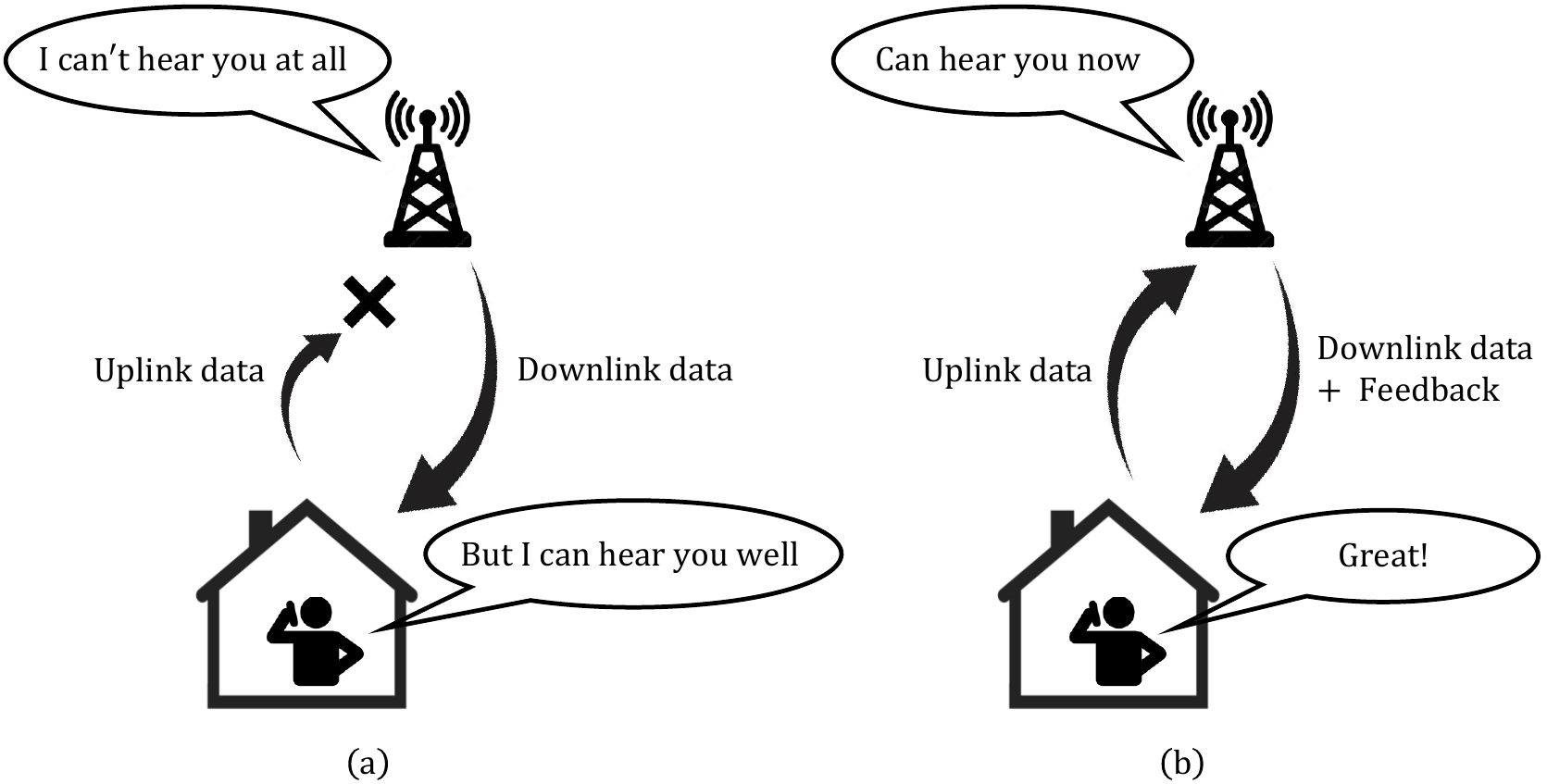}\\
	\caption{(a) Typical asymmetric communication scenario, where downlink is successful while uplink fails. (b) Enhanced uplink with real-time feedback, where IoT devices leverage feedback to improve uplink coverage.}
	\label{fig:Introduction}
\end{figure}

To mitigate such uplink coverage limitations, conventional strategies have involved increasing uplink transmission power, employing repetitive transmissions (e.g., NB-IoT \cite{kanj2020tutorial}), or exploiting multi-antenna diversity (e.g., the Transmission Mode (TM) 2 of LTE and TM 9 of 5G NR \cite{parkvall2017nr}). However, these methods are not favorable or impractical for energy-constrained, single-antenna IoT devices. 

In this context, we seek an alternative that boosts uplink connectivity without raising the transmission power of IoT devices. Our proposed solution is a feedback-aided deep and wide coverage communication approach, built on the innate asymmetry between downlink and uplink communications of IoT cells. Specifically, APs typically have multiple antennas and virtually unlimited power resources, allowing them to achieve extensive downlink coverage, as shown in Fig.~\ref{fig:Introduction}(a). By exploiting this capability, the AP can transmit real-time feedback -- indicating its current decoding status -- to the IoT device, thereby boosting the IoT device's uplink coding efficiency through feedback channel codes \cite{Shannon,Polyanskiy}. As illustrated in Fig.~\ref{fig:Introduction}(b), this feedback mechanism effectively extends the uplink coverage radius at the same device power consumption, thereby connecting the unconnectable.

The cornerstone of our proposed approach lies in the innovative use of feedback channel codes \cite{Shannon,Polyanskiy,shao2024deep,SK1,GBAFC}. Unlike traditional forward-only error correction codes, feedback channel codes utilize real-time feedback from the AP to dynamically adjust the coding strategy and address mis-decoding at the receiver. This capability enables the receiver to decode data at much lower SNR levels, effectively extending the communication coverage. The introduction of feedback brings a dual dependency: uplink decoding performance is now influenced by both the uplink and downlink channels, a phenomenon we term dual-channel coupling. This concept is related to rateless and adaptive feedback coding schemes \cite{luby2002lt,shokrollahi2006raptor}, where the transmitter continuously sends coded symbols until successful decoding is confirmed. These schemes exploit feedback to adaptively allocate redundancy, enabling reliable communication even under uncertain or poor channel conditions.

The contributions of this paper are summarized as follows.
\begin{itemize}[leftmargin=0.4cm]
    \item We introduce a new uplink coverage extension approach that harnesses AP's real-time feedback to enhance the coverage probability in uplink-downlink asymmetric communications, a typical scenario in mobile edge networks.
    \item We quantitatively analyze the uplink coverage probabilities and the resultant number of connectable APs under the feedback-aided coverage extension paradigm. By solving the uplink-downlink dual-channel coupling, we show that real-time feedback effectively mitigates the exponential decay in coverage probability with distance, enabling the IoT device to establish links with more APs -- even those previously out of reach -- without increasing transmit power.
    \item Our results show that the coverage improvement is most pronounced for devices at the cell edge, where uplink SNR is weakest. By significantly boosting coverage probability in these challenging regions, the proposed feedback-based approach offers a robust pathway to ``connect the unconnectable'', thereby enabling dependable IoT services in demanding environments.
\end{itemize}


\begin{figure}
    \centering
        \centering        \includegraphics[width=0.7\columnwidth]{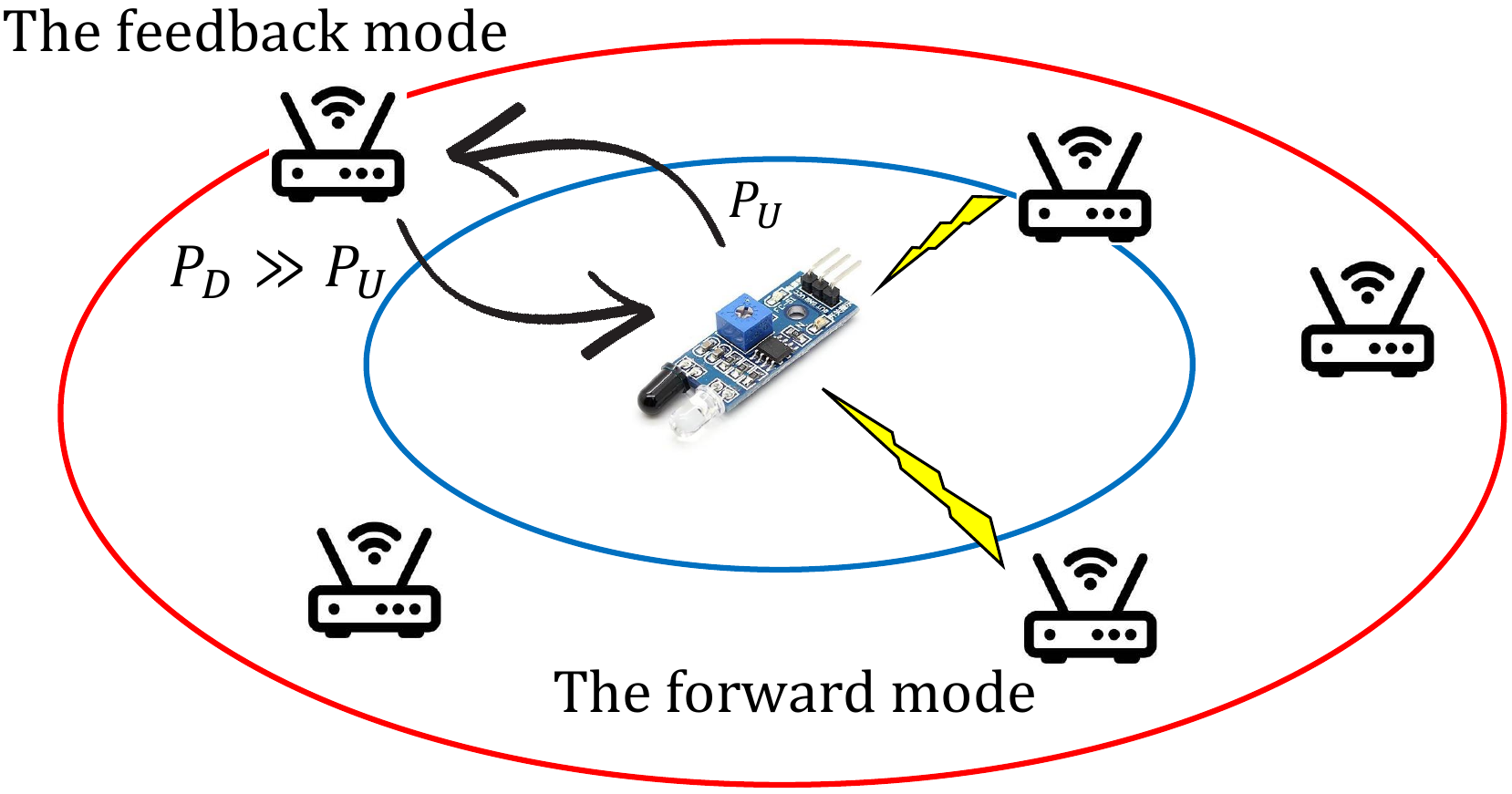}\\
	\caption{An IoT Device connects with distributed APs. The target scenario of our feedback-aided coverage extension scheme includes both fixed and mobile IoT devices.}
\end{figure}

\section{System Model}\label{sec:model}
We consider an IoT environment comprising distributed wireless access points (APs). Focusing on a single IoT device, our primary interest lies in its uplink transmission coverage, specifically, the number of APs that fall within the communication range of the device with the real-time feedback from the APs. 

We model the spatial distribution of the APs using a Poisson Point Process (PPP) with intensity $\lambda$, representing the average number of APs per unit area. The uplink path loss from the IoT device to an AP at a distance $R$ can be written as
\begin{equation}\label{def:1}
L_{U}(R)=C_{U}R^{-\alpha_{U}},
\end{equation}
where $C_{U}$ is the uplink path loss intercepts and $\alpha_{U}$ is the uplink path loss exponents \cite{shao2021federated}.

The transmit power of the IoT device and the AP are denoted by $P_U$ and $P_D$, respectively. Due to energy constraints typical in IoT devices, we have $P_D\gg P_U$. The uplink signal-to-noise ratio (SNR) at the AP can be expressed as
\begin{equation}\label{eq:snru}
\eta_U(R)=P_U\frac{G_{t}G_{r}L_{U}(R)}{\sigma_{U}^2}|h_U|^{2},
\end{equation}
where $G_t$ and $G_r$ denote the transmit and receive antenna gains, respectively; $\sigma_{U}^2$ is the additive white Gaussian noise (AWGN) power; $h_U$ is the uplink small-scale fading coefficient. We model $h_U$ by Rayleigh distribution, i.e., {$|h_{U}|^{2}\sim f(|h_{U}|^{2};\mu_U)$, where $f(x)$ denotes the probability density function (PDF) of the exponential distribution.} In this paper, we will also write SNR in decibels, in which case a `dB' will be added in the subscript, e.g., $\eta_{U,\text{dB}}\triangleq 10\lg \eta_{U}$.

Downlink path loss is similarly modeled but with different parameters $L_{D}(R)\!=\!C_{D}R^{-\alpha_{D}}$, where $\alpha_{D}$, $C_{D}$ are the respective downlink path loss exponents and intercepts.
The downlink SNR received at the IoT device can be written as
\begin{equation}\label{def:d}
\eta_D(R)=P_D\frac{G_{t}G_{r}L_{D}(R)}{\sigma_{D}^2}|h_D|^{2},
\end{equation}
where $\sigma_{D}^2$ is AWGN power; $h_D$ is the downlink Rayleigh fading coefficient, and {$|h_{D}|^{2} \sim f(|h_{D}|^{2};\mu_D)$.}

To transmit a packet of $K$ bits, the IoT device conventionally employs a forward error correction code $\mathcal{C}(K,N,\epsilon^*)$ to protect the information bits, where $N$ is the number of uplink channel uses, and $\epsilon^*$ is the target packet error rate (PER), reflecting the target throughput. By the finite length coding theorem \cite{polyanskiy2010channel}, the PER $\epsilon$ is determined by the uplink SNR $\eta_{U}$:
\begin{equation}
    \sqrt{NV} Q^{-1}(\epsilon) \approx \frac{N}{2}\log (1+\eta_{U})-K,
\end{equation}
where $V=\frac{\eta_{U}(\eta_{U}+2)}{2(\eta_{U}+1)^2}\log^2 e$ is the channel dispersion, and $Q$ represents the Q-function. 

We refer to this conventional approach, where only forward channel coding is employed, as the forward mode. In forward mode, the critical uplink SNR required to achieve the target PER $\epsilon^*$ is denoted by $\Omega_c$. Therefore, an AP is said to be connectable for the IoT device, or the IoT device is under the coverage of the AP, if the uplink SNR $\eta_{U}\geq \Omega_c$.

When the AP provides real-time feedback to the IoT device, the channel coding efficiency can be significantly improved. This enhancement effectively lowers the critical uplink SNR required to achieve the target PER, thereby increasing the probability that an AP at any distance $R$ is connectable. We refer to this enhanced communication approach as the feedback mode.

In feedback mode, the IoT device employs a feedback error correction code $\mathcal{C}_f(K,N,N',\epsilon^*)$, where $K$, $N$, $\epsilon^*$ are as defined in $\mathcal{C}(K,N,\epsilon^*)$, while $N'$ is the number of downlink channel uses dedicated to feedback. 
In this paper, we consider the feedback coding architecture in \cite{shao2024deep,GBAFC}, in which case $N'=aN$, where $a\in\mathbb{Z}^+$ is a positive integer. The critical uplink SNR in the feedback mode required to meet the target $\epsilon^*$, denoted as $\Omega_{f}$, follows a logistic function:
\begin{equation}\label{eq:fit}
    \Omega_{f,\text{dB}} = \frac{1}{ \exp\left(u_0{\eta}_{D,\text{dB}}+u_1 a+u_2{\eta}_{D,\text{dB}} a + u_3\right) + u_4} + u_5,
\end{equation}
where $\{u_0,u_1,u_2,u_3,u_4,u_5\}$ are constants, determined by the target PER $\epsilon^*$. An AP is deemed connectable for the IoT device if the uplink SNR $\eta_{U}\geq \Omega_f$. 

As can be seen from \eqref{eq:fit}, unlike the critical SNR $\Omega_c$ in the forward mode, which depends solely on the target PER $\epsilon^*$, the critical SNR $\Omega_f$ in the feedback mode also depends on the feedback channel quality ${\eta}_{D}$ and allocated feedback channel resources $a$.
Consequently, $\Omega_f$ becomes a random variable due to the inherent randomness of the downlink channel quality ${\eta}_{D}$. This dual dependency -- on both uplink and downlink channel conditions -- introduces a unique coupling effect between the two channels. Specifically, the uplink channel determines the IoT device's ability to meet the SNR threshold for a given distance to the AP, while the downlink channel impacts the efficiency of feedback communication, which in turn affects the required critical SNR in the uplink.

\section{Uplink Coverage Analysis with Feedback}\label{sec:coverage_analysis}
In this section, we analyze the number of APs that are connectable in the feedback mode and quantify the performance gains achieved by leveraging real-time feedback. Let us start with the coverage probability analysis.

\begin{prop}\label{prop:1}
In feedback mode, the probability that an AP at a distance $R$ is connectable, denoted by $\varphi_{f}(R)\triangleq\Pr\big(\eta_{U}(R)\geq\Omega_f\big)$, can be approximated as
\begin{equation*}
\varphi_{f}(R) \approx\sum\limits_{k=1}^{{L}}w_{k}e^{-\mu_{U}\left(J_{1,k}R^{\alpha_{U}}-J_{2,k}\right)},
\end{equation*}
where the approximation is based on a Gauss-Laguerre quadrature of order ${L}$, with $x_k$ denoting the roots of the Laguerre polynomials and $w_k$ the corresponding weights \cite{ioakimidis1993gauss}. The coefficients $J_{1,k}$ and $J_{2,k}$ are functions of $x_k$, as defined in \eqref{eq:26}.
\end{prop}

\begin{proof}
In feedback mode, the interplay between the uplink and downlink channels introduces a coupled dual-channel fading behavior. This coupling complicates the analysis of feedback communication systems, as the statistical properties of both the uplink and downlink channels must be jointly considered. Accurately characterizing this interdependence is critical for evaluating the system's coverage performance.

For the probability that an AP at a distance $R$ is connectable $\varphi_{f}(R)\triangleq\Pr\big(\eta_{U}(R)\geq\Omega_f\big)$, both $\eta_{U}(R)$ and $\Omega_f$ are random variables: the randomness of $\eta_{U}(R)$ arises due to uplink small-scale fading, while the randomness of $\Omega_f$ is introduced through the downlink SNR $\eta_{D}(R)$, which influences the feedback quality.

Substituting the expression for $\eta_{U}(R)$ into $\varphi_{f}(R)$, the coverage probability can be reformulated as
\begin{eqnarray}\label{eq:6}
&&\hspace{-0.5cm}
\varphi_{f}(R)= \Pr\Big(|h_{U}|^{2}>\frac{{\Omega_f}\sigma_{U}^2}{P_{U}G_{t}G_{r}C_{U}R^{-\alpha_{U}}}\Big).
\end{eqnarray}
To simplify the notation, we introduce a function $g(R,|h_{D}|^2)$:
\begin{equation*}
g(R,|h_{D}|^2)\triangleq \frac{{\Omega}_{f}\sigma_{U}^{2}}{P_{U}G_{t}G_{r}C_{U}R^{-\alpha_{U}}},
\end{equation*}
where $\Omega_{f}$, as given in \eqref{eq:fit}, is governed by the downlink received SNR $\eta_{D}$. 

Since both $|h_{U}|^{2}$ and $|h_{D}|^{2}$ follow exponential distributions, the coverage probability can be rewritten as
\begin{eqnarray}\label{eq:7}
&&\hspace{-0.5cm}
\varphi_{f}(R)= \Pr\Big(|h_{U}|^{2}>g(R,|h_{D}|^2)\Big)\notag\\
&&\hspace{-0.5cm}=\int_{0}^{+\infty}[1-(1-e^{-\mu_{U}g(R,|h_{D}|^2)})]f(|h_{D}|^2)d|h_{D}|^2\notag\\
&&\hspace{-0.5cm}=\int_{0}^{+\infty}e^{-\mu_{U}g(R,x)}\mu_{D}e^{-\mu_{D}x}dx\notag\\
&&\hspace{-0.5cm}\overset{(a)}\approx\sum\limits_{k=1}^{{L}}w_{k}e^{-\mu_{U}g\left(R,\frac{x_{k}}{\mu_{D}}\right)},
\end{eqnarray}
where (a) follows from the Gauss-Laguerre quadrature of order $L$, with $x_k$ denoting the roots of the Laguerre polynomials and $w_k$ the corresponding weights. A higher value of $L$ generally leads to more accurate approximations.

To derive the coverage probability, an accurate approximation of $g(R,|h_{D}|^2)$ is necessary. We start by approximating the downlink SNR in decibels, denoted as $\eta_{D,dB}\left(\frac{x_{k}}{\mu_{D}}\right)$, as follows:
\begin{eqnarray}\label{17}
&&\hspace{-0.5cm}\eta_{D,dB}\left(\frac{x_{k}}{\mu_{D}}\right)=10\lg{\left(\frac{P_{D}G_{t}G_{r}C_{D}x_{k}R^{-\alpha_{D}}}{\mu_{D}\sigma_{D}^2}\right)}\notag\\
&&\hspace{-0.5cm}=10\lg{(Z_{1,k}R^{-\alpha_{D}})}
=10\lg{Z_{1,k}}+10\lg{R^{-\alpha_{D}}}\notag\\
&&\hspace{-0.5cm}=10\lg{Z_{1,k}}+\frac{10}{\ln10}\ln{R^{-\alpha_{D}}}\notag\\
&&\hspace{-0.5cm}\approx 10\lg{Z_{1,k}}+\frac{10}{\ln10}(R^{-\alpha_{D}}-1),
\end{eqnarray}
where $Z_{1,k}\triangleq\frac{P_{D}G_{t}G_{r}C_{D}x_{k}}{\mu_{D}\sigma_{D}^{2}}$ and $Z_{2}\triangleq\frac{\sigma_{U}^{2}}{P_{U}G_{t}G_{r}C_{U}}$ encapsulates system parameters for the downlink and uplink, respectively.

Substituting \eqref{17} into \eqref{eq:fit}, the critical SNR in decibels $\Omega_{f,dB}$ can be approximated as
\begin{equation}
\Omega_{f,dB}\approx\frac{1}{\exp{(B_{1,k}+B_{2}R^{-\alpha_{D}}-1-u_4})+u_4}+u_5,
\end{equation}
where $B_{1,k}=1+u_{1}a+u_{3}+u_{4}+10(u_0+u_{2}a)(\lg{Z_{1,k}}-\frac{1}{\ln10})$ and $B_{2}=\frac{10}{\ln10}(u_{0}+u_{2}a)$.

Using the first-order Taylor series expansion $e^x=1+x+R_{1}^{(1)}$, the approximation of $\Omega_{f,dB}$ becomes
\begin{equation}
\Omega_{f,dB}\approx\frac{1}{B_{1,k}+B_{2}R^{-\alpha_{D}}+R_{1}^{(1)}}+u_{5}.
\end{equation}
Applying another Taylor series expansion $\frac{1}{1+x}=1-x+R_{1}^{(2)}$, we further simplify $\Omega_{f,dB}$ as
\begin{equation}\label{eq:omega}
\Omega_{f,dB}\approx\frac{1}{B_{1,k}}\left(1-\frac{B_{2}}{B_{1,k}}R^{-\alpha_{D}}-\frac{R_{1}^{(1)}}{B_{1,k}}+R_{1}^{(2)}+{u_{5}}{B_{1,k}}\right).
\end{equation}
Here, the bounds of the Lagrange remainder terms satisfy $R_{1}^{(1)}\geq 0$ and $R_{1}^{(2)}\geq 0$.

We now manipulate $g\left(R,\frac{x_{k}}{\mu_{D}}\right)$, defined as the effective threshold function:
\begin{eqnarray}\label{eq:g}
&&\hspace{-0.5cm}g\left(R,\frac{x_{k}}{\mu_{D}}\right)=Z_{2}R^{\alpha_{U}}10^{\frac{1}{10}\Omega_{f,dB}}\\
&&\hspace{-0.5cm}=Z_{2}R^{\alpha_{U}}e^{\frac{\ln10}{10}\Omega_{f,dB}}
= Z_{2}R^{\alpha_{U}}\left(1+{\frac{\ln10}{10}\Omega_{f,dB}}+R_{1}^{(3)}\right), \notag
\end{eqnarray}
where $R_{1}^{(3)}$ represents the remainder of the Taylor expansion.

Substituting \eqref{eq:omega} into \eqref{eq:g}, we obtain
\begin{eqnarray}\label{eq:24}
&&\hspace{-0.5cm}g\left(R,\frac{x_{k}}{\mu_{D}}\right)\approx Z_{2}R^{\alpha_{U}}\bigg[1+\frac{\ln{10}}{10B_{1,k}}\Big(1-\frac{B_{2}}{B_{1,k}}R^{-\alpha_{D}}\notag\\
&&\hspace{-0.5cm}\qquad\qquad\qquad -\frac{R_{1}^{(1)}}{B_{1,k}}+R_{1}^{(2)}+{u_{5}}{B_{1,k}}\Big)+R_{1}^{(3)}\bigg].
\end{eqnarray}
In particular, the Lagrange form of the remainder $R_{1}^{(3)}$ can be approximated by
\begin{eqnarray}
&&\hspace{-0.5cm}R_{1}^{(3)}=\frac{1}{2!}e^{\frac{\ln{10}}{10B_{1,k}}(1-\frac{B_{2}}{B_{1,k}}c^{-\alpha_{D}}-\frac{R_{1}^{(1)}}{B_{1,k}}+R_{1}^{(2)}
+u_{5}B_{1,k})}\\
&&\hspace{-0.5cm}\times\bigg[{\frac{\ln{10}}{10B_{1,k}}(1-\frac{B_{2}}{B_{1,k}}R^{-\alpha_{D}}-\frac{R_{1}^{(1)}}{B_{1,k}}+R_{1}^{(2)}
+u_{5}B_{1,k})}\bigg]^2\notag\\
&&\hspace{-0.5cm}\approx\frac{1}{2}\bigg[{\frac{\ln{10}}{10B_{1,k}}(1\!-\!\frac{B_{2}}{B_{1,k}}R^{-\alpha_{D}}\!-\!\frac{R_{1}^{(1)}}{B_{1,k}}\!+\!R_{1}^{(2)}
\!+\!u_{5}B_{1,k})}\bigg]^2.\notag
	\label{eq:R}
\end{eqnarray}
Since $R_{1}^{(1)}\geq 0$, $R_{1}^{(2)}\geq 0$, and $(1+x)^\alpha\approx{1+\alpha{x}}$, $R_{1}^{(3)}$ can be further refined as
\begin{eqnarray}
&&\hspace{-0.5cm}R_{1}^{(3)}\approx \frac{\ln^2{10}}{200B_{1,k}^2}(1+u_{5}B_{1,k})^2\bigg[1-\frac{B_{2}R^{-\alpha_{D}}}{B_{1,k}(1+u_{5}B_{1,k})}\bigg]^2\notag\\
&&\hspace{-0.5cm}\approx \frac{\ln^2{10}}{200B_{1,k}^2}(1+u_{5}B_{1,k})^2\bigg[1-\frac{2B_{2}R^{-\alpha_{D}}}{B_{1,k}(1+u_{5}B_{1,k})}\bigg].
	\label{eq:25}
\end{eqnarray}

Substituting \eqref{eq:25} into \eqref{eq:24} yields
\begin{equation}\label{eq:26}
g\left(R,\frac{x_{k}}{\mu_{D}}\right)\approx J_{1,k}R^{\alpha_{U}}-J_{2,k},
\end{equation}
where $J_{1,k}\triangleq{Z_2}+\frac{Z_{2}(1+u_{5}B_{1,k})\ln10}{10B_{1,k}}+\frac{Z_{2}(1+u_{5}B_{1,k})^{2}\ln^{2}10}{200B_{1,k}^2}$ and $J_{2,k}\triangleq\frac{Z_{2}B_{2}\ln10}{10B_{1,k}^2}+\frac{Z_{2}B_{2}(1+u_{5}B_{1,k})\ln^2{10}}{100B_{1,k}^3}$. Finally, $\varphi_{f}(R)$ in Proposition \ref{prop:1} can be obtained by substituting \eqref{eq:26} into \eqref{eq:7}.
\end{proof}

To provide a benchmark for comparison with the feedback mode, we derive $\varphi_{c}(R)$, the probability that a AP at a distance $R$ is connectable in forward mode, where no feedback is employed. Specifically, we have
\begin{equation}\label{eq:13}
\varphi_{c}(R) = \Pr\big(\eta_{U}(R)\geq\Omega_c\big)=e^{-AR^{\alpha_{U}}},
\end{equation}
where $A\triangleq\frac{\mu_{U}{\Omega_c}\sigma_{U}^2}{P_{U}G_{t}G_{r}C_{U}}$. 

The comparison between the coverage probabilities $\varphi_{c}(R)$ in forward mode and $\varphi_{f}(R)$ in feedback mode highlights the substantial improvement brought by feedback.

In forward mode, the coverage probability decays exponentially with distance $R$, as captured by the fixed term $e^{-AR^{\alpha_{U}}}$, which depends solely on the uplink channel quality. In contrast, feedback mode leverages both uplink and downlink channels, where the feedback gains are encapsulated in the terms $J_{1,k}$ and $J_{2,k}$:
\begin{itemize}
    \item The term $J_{1,k}$ controls the rate at which coverage probability decays with distance. A smaller $J_{1,k}$ implies a slower decay, i.e., better coverage at larger distances.
    \item The term $J_{2,k}$ provides an additive shift gain that further boosts coverage, especially in low-SNR regions.
\end{itemize}

These two terms are directly influenced by the feedback parameters $a$ (feedback channel usage) and $\eta_D$ (feedback SNR). In particular, increasing $a$ and/or having a higher $\eta_D$ increases the feedback quality, effectively reducing the uplink SNR threshold $\Omega_f$. This leads to lower $J_{1,k}$ and higher $J_{2,k}$, thus remarkably improving coverage probability.
The improvement is particularly pronounced at larger distances, such as at the cell edge, where  the uplink SNR is typically weak, and connectivity is more prone to failure. By increasing the coverage probability through feedback, the connectivity in these edge regions can be significantly enhanced.

Given the coverage probability analysis, we next derive the number of connectable APs in both the forward and feedback modes.
We shall focus on a circular coverage area centered around the IoT device with radius $D$. 

\begin{prop}\label{prop:2}
Let $M_c(D)$ and $M_f(D)$ denote the number of APs that are connectable by an IoT device within a circular region of radius $D$ in the forward and feedback modes, respectively. Then, we have
  \begin{eqnarray}
  &&\hspace{-1cm} M_{f}(D)=2\pi\lambda\sum\limits_{k=1}^{{L}}w_{k}e^{\mu_{U}J_{2,k}}\frac{\gamma(\frac{2}{\alpha_{U}},\mu_{U}J_{1,k}{D}^{\alpha_{U}})}{(\mu_{U}J_{1,k})^{\frac{2}{\alpha_{U}}}\alpha_{U}}. \\
  &&\hspace{-1cm} M_{c}(D) =\frac{2\pi\lambda\cdot\gamma(\frac{2}{\alpha_{U}},AD^{\alpha_{U}})}{\alpha_{U}A^{\frac{2}{\alpha_{U}}}}.
\end{eqnarray}
where $\gamma(\cdot,\cdot)$ denotes the lower incomplete gamma function.
\end{prop}

\begin{proof}
In the feedback mode, given the coverage probability $\varphi_{f}(R)$, the number of connectable APs $M_{f}(D)$ within a circular area of radius $D$ can be expressed as
\begin{eqnarray}
&&\hspace{-0.5cm}M_{f}(D)=\int_{0}^{D}\varphi_{f}(R)\lambda2\pi{R}dR\notag\\
&&\hspace{-0.5cm}\approx\int_{0}^{D}\sum\limits_{k=1}^{{L}}w_{k}e^{-\mu_{1}g\left(R,\frac{x_{k}}{\mu_{D}}\right)}\lambda2\pi{R}dR\notag\\
&&\hspace{-0.5cm}=2\pi\lambda\sum\limits_{k=1}^{{L}}w_{k}\int_{0}^{D}e^{-\mu_{1}g\left(R,\frac{x_{k}}{\mu_{D}}\right)}{R}dR,
	\label{def:25}
\end{eqnarray}

Substituting the approximation of $g\left(R,\frac{x_{k}}{\mu_{D}}\right)$ in \eqref{eq:26} yields
\begin{eqnarray}
&&\hspace{-0.5cm}M_{f}(D)\approx2\pi\lambda\sum\limits_{k=1}^{{L}}w_{k}e^{\mu_{U}J_{2,k}}\int_{0}^{D}e^{-\mu_{U}J_{1,k}R^{\alpha_{U}}}{R}dR\notag\\
&&\hspace{-0.5cm}=2\pi\lambda\sum\limits_{k=1}^{{L}}w_{k}e^{\mu_{U}J_{2,k}}\frac{\gamma(\frac{2}{\alpha_{U}},\mu_{U}J_{1,k}{D}^{\alpha_{U}})}{(\mu_{U}J_{1,k})^{\frac{2}{\alpha_{U}}}\alpha_{U}}.
\end{eqnarray}
where $\gamma(\cdot,\cdot)$ is lower incomplete gamma function.

On the other hand, in the forward mode, the number of connectable APs $M_{c}(D)$ within a circular area of radius $D$ can be similarly expressed as
\begin{equation*}
M_{c}(D)= \int_{0}^{D}\varphi_{c}(R)\lambda2\pi{R}dR =\frac{2\pi\lambda\cdot\gamma(\frac{2}{\alpha_{U}},AD^{\alpha_{U}})}{\alpha_{U}A^{\frac{2}{\alpha_{U}}}}.
\end{equation*}
\end{proof}

\begin{rem}
    The analytical results in Propositions \ref{prop:1} and \ref{prop:2} rely on a series of approximations. To clarify the general validity of these approximations, we provide a comprehensive sensitivity analysis in Appendix \ref{sec:AppA}, examining how the accuracy of the analytical expressions varies with key system parameters.
\end{rem}

\section{Numerical Results}\label{sec:numericalresults}
This section evaluates the performance of the proposed feedback-aided coverage enhancement approach through numerical simulations. {To establish an evaluation framework, we configure the system parameters as follows. The density of APs is set to $\lambda=6\times10^{-3}$, with intercepts and path loss exponents set as $C_U\!=\!C_D\!=\!10^{-4.7}$ and $\alpha_U\!=\!\alpha_D\!=\!4$ \cite{mmWave}. 
The uplink and downlink Rayleigh fading channels are modeled with parameters $\mu_U=\mu_D=2$.
We consider a short block length $K=48$, a code rate of $1/3$, resulting in $N=144$ channel uses, and a target uplink PER $\epsilon^*=10^{-4}$. 
The downlink transmit power is fixed at $P_D=50$ mW, while the uplink transmit power is varied across $P_U\in [0.5, 1, 2]$ mW to examine the impact of different power levels on system performance. }
In the forward mode, Polar and Turbo codes with the same block length and coding rate are utilized as benchmarks, enabling a direct comparison of their performance against the proposed feedback-aided approach.

\begin{rem}
For Turbo codes, we followed the LTE standard. The encoder structure and interleaver pattern conform to the LTE specifications, and the decoder employs a Max-Log-MAP algorithm consistent with LTE decoding procedures. This setup ensures that the benchmark reflects practical performance achievable in standardized systems.

For Polar codes, we adopted a configuration aligned with the 5G New Radio (NR) standard. The code construction is based on the 5G polar design with a 1-bit parity CRC. Decoding is performed using a CRC-aided Successive Cancellation List (SCL) decoder with a list size of 4. This configuration provides a strong benchmark for short block lengths, as considered in our evaluation.
\end{rem}

\begin{figure}[t]
    \centering
    \begin{subfigure}{.38\textwidth}
        \centering
        \includegraphics[width=\textwidth]{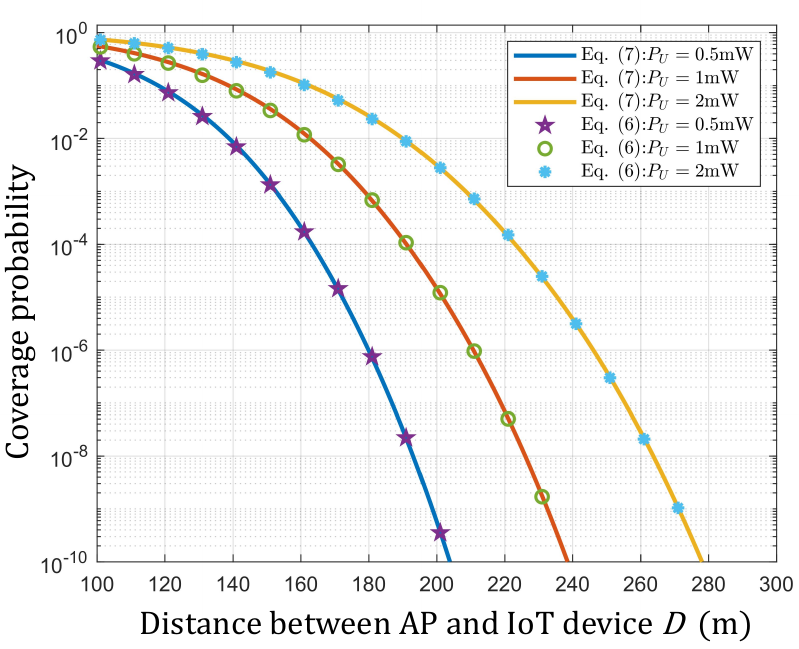}
    \end{subfigure}
    \begin{subfigure}{.38\textwidth}
        \centering
        \includegraphics[width=\textwidth]{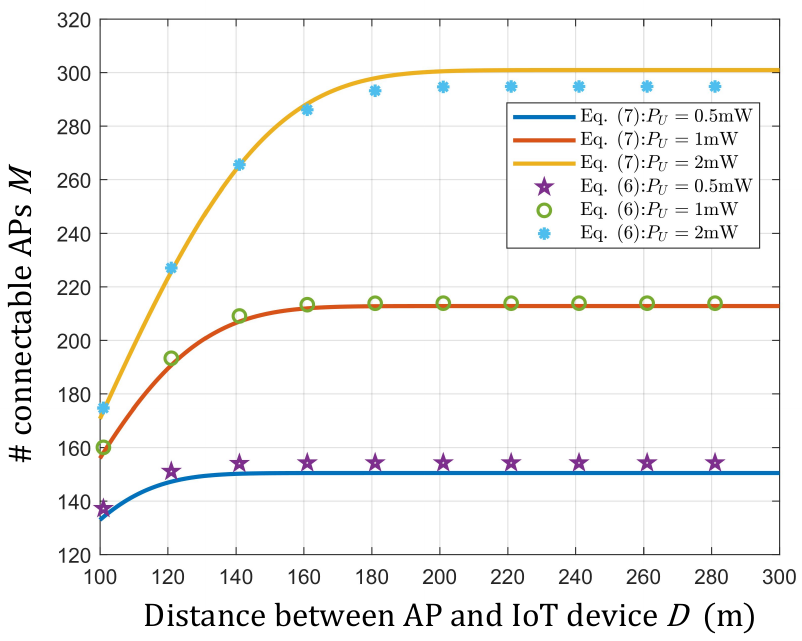}
    \end{subfigure}
    \caption{Comparison of analytical and simulation results to validate the coverage analysis in the feedback mode with the DEEP-IoT feedback coding scheme.}
    \label{fig:S0}
\end{figure}

We begin by verifying the accuracy of the analytical models established for feedback-aided connectivity in Section \ref{sec:coverage_analysis}. 
To this end, we compare the derived expressions for the number of connectable APs with simulation results, as illustrated in Fig.~\ref{fig:S0}. This comparison focuses on two key steps: the approximations of coverage probability and the computation of connectable APs. Later in Figs.~\ref{fig:a_probability} and \ref{fig:a_number}, we will introduce practical Polar and Turbo codes (in the forward mode) to benchmark coverage performance against the DEEP-IoT scheme (in the feedback mode).
The results in Fig.~\ref{fig:S0} reveal a strong alignment between the analytical predictions and simulation outcomes across varying uplink transmit power levels and distances to the IoT device.

\begin{figure}
	\centering
	\includegraphics[width=0.38\textwidth]{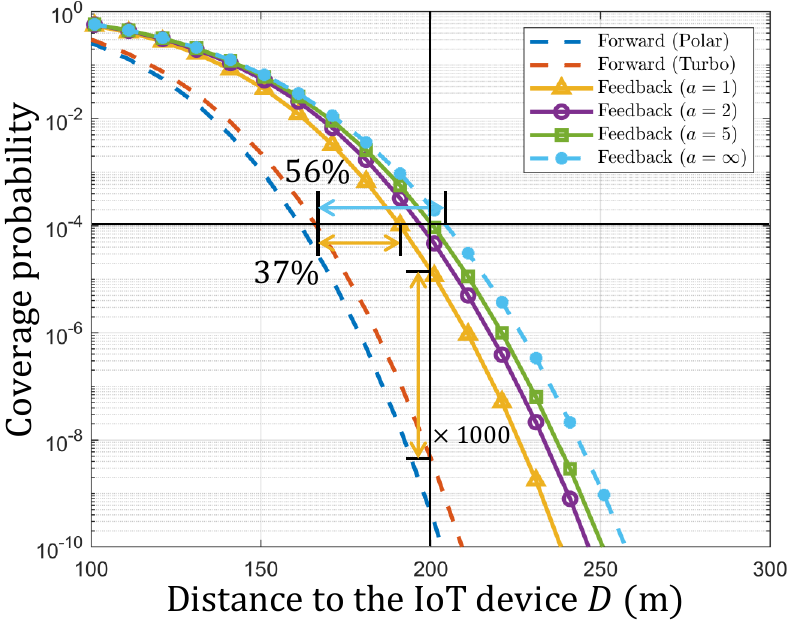}\\
	\caption{The coverage probability gains of the feedback mode versus the forward mode with Polar and Turbo codes.}
	\label{fig:a_probability}
\end{figure}

\begin{figure}
	\centering
	\includegraphics[width=0.38\textwidth]{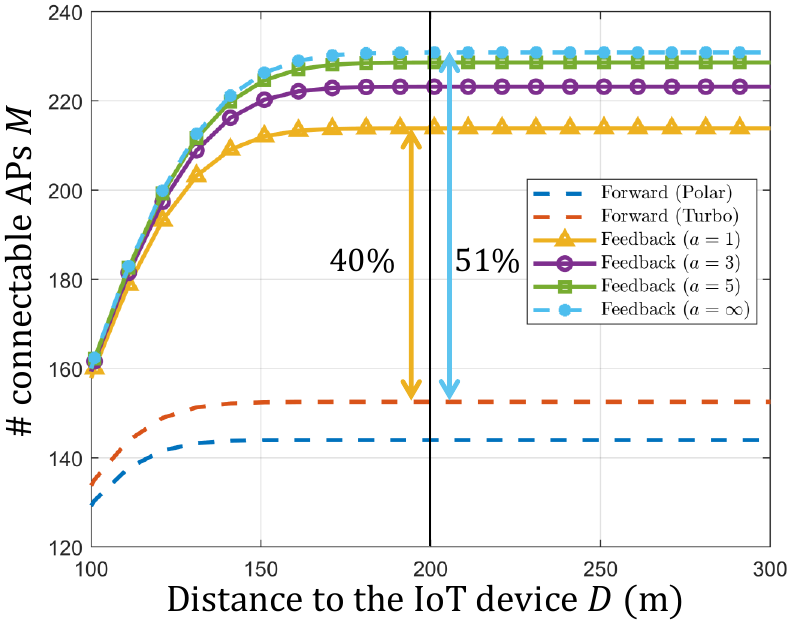}\\
	\caption{The gains in the number of connectable APs achieved by the feedback mode versus the forward mode using Polar and Turbo codes.}
	\label{fig:a_number}
\end{figure}

Next, we explore the transformative impact of feedback on connecting previously unreachable APs. Fig.~\ref{fig:a_probability} illustrates the substantial enhancement in coverage probability achieved through feedback at varying distances. Notably, its benefits are most pronounced at larger distances, particularly for IoT devices situated at the cell edge. At $200$ meters, for instance, feedback boosts the coverage probability by three orders of magnitude, effectively bridging the gap where traditional forward-only communication fails.
The extent of this improvement is closely tied to the amount of feedback, $a$. More feedback allows the user to precisely assess the decoding status of the AP, significantly extending the coverage range. Without feedback, Turbo codes ensure a coverage probability of $10^{-4}$ up to $166$ meters. In stark contrast, feedback extends this reliability to $205$ meters -- a remarkable $24\%$ increase in range, enabling connections far beyond the traditional limits.

The benefits of feedback extend beyond probability gains. Fig.~\ref{fig:a_number} showcases the number of connectable APs, $M$, at varying distances. At a distance of $200$ meters, feedback enhances the number of reachable APs by $51\%$, transforming what was previously an unconnected region into a significantly more connected zone. These results emphasize the power of feedback in ``connecting the unconnectable'', redefining the boundaries of IoT network performance.

\section{Conclusions}
\label{sec:conclusions}
This work has laid the groundwork for a transformative approach to enhancing IoT uplink connectivity through feedback-aided communication. By introducing and analyzing the dual-channel coupling effect, we have shown that real-time feedback fundamentally redefines the coverage possibilities in edge scenarios. Our findings underscore the strategic use of feedback to counteract the asymmetry between uplink and downlink performance, bridging what was previously an ``unconnectable'' gap.

While the proposed approach significantly improves uplink coverage, especially at the cell edge, we acknowledge that this gain comes with a cost: increased downlink bandwidth consumption and latency due to the feedback process. The bandwidth overhead is explicitly modeled via the number of downlink channel uses. As for latency, it is less critical in many IoT applications, such as NB-IoT and LTE-M, where delay tolerance is high and the benefits of improved coverage outweigh the delay introduced by feedback.

\appendices

\section{Sensitive Analysis of Approximations}\label{sec:AppA}
Our derivations in Propositions \ref{prop:1} and \ref{prop:2} rely on a series of mathematical approximations. To clarify the validity and practical applicability of these approximations, we now provide a detailed analysis of the approximation error in relation to key system parameters.

\vspace{0.3em} 
\noindent\textbf{Step 1: Approximating $\boldsymbol{\Omega_{f,\text{dB}}}$ using $\boldsymbol{e^x \approx 1+x}$ (Eq.(9) to Eq.(10))}

To analyze the accuracy of the first approximation, we define
\begin{eqnarray}\label{eq:yaR}
&&\hspace{-0.6cm} y(a,R) \triangleq B_{1,k} + B_2 R^{-\alpha_D} - 1 - u_4,  \\
&&\hspace{-0.6cm} B_{1,k}=1\!+\!u_{1}a\!+\!u_{3}\!+\!u_{4}\!+\!10(u_0\!+\!u_{2}a)(\lg{Z_{1,k}}\!-\!\frac{1}{\ln10}), \notag\\
&&\hspace{-0.6cm} B_{2}=\frac{10}{\ln10}(u_{0}+u_{2}a). \notag
\end{eqnarray}

We first analyze how the error behaves with respect to the number of feedback channel uses $a$. For any fixed distance $R$, $y(a, R)$ is a monotonically increasing function of $a$. The absolute error introduced by approximating $e^y$ as $1 + y$ is
\begin{equation}
    \delta(y)=\frac{1}{y+1+u_4}-\frac{1}{e^{y}+u_{4}}.
\end{equation}
Since $a\geq 1$, we have $y(a,R) \geq y(1,R) >   u_1+u_3+10(u_0+u_2)(\lg{Z_{1,k}-\frac{1}{\ln10})} \triangleq \Gamma$. 

Taking the derivative of $\delta(y)$ with respect to $y$, we obtain
\begin{equation}\label{eq:1}
    \delta'(y) = \frac{e^{y}[(y+1)^2-e^y+2u_4]-u_4^2}{(e^y+u_4)^2(y+1+u_4)^2}.
\end{equation}

Since the denominator is always positive, the sign of $\delta'(y)$ is determined by the numerator, which we define as
\begin{equation}
    s(y)\triangleq(y+1)^2-e^y+2u_4.
\end{equation}
We analyze $s(y)$ and its derivatives:
\begin{equation}
    s'(y)=2y+2-e^y,~~s''(y)=2-e^y.
\end{equation}

For $a\geq 1$, we typically have $s''(y) <0$, hence $s'(y)$ is monotonically decreasing. Furthermore, from $s'(y(1,R))<0$, we have $s'(y)<0$ and $s(y)$ is strictly decreasing; from $s(y(1,R))<0$, we have $s(y)<0$. 

As a result, the derivative in \eqref{eq:1} $\delta'(y)<0$. The absolute error $\delta(y)$ monotonically decreases as $y$ increases. Due to the nature of composite functions, $\delta(y)$ is a monotonically decreasing function of $y$, and consequently, of $a$.

We next analyze the impact of distance $R$ on the approximation. From \eqref{eq:yaR}, we know that $y(a,R)$ is monotonically decreasing in $R$. Since the absolute error $\delta(y)$ decreases as $y$ increases, it follows that $\delta(y)$ is a monotonically increasing function of $R$.

\textit{Summary:} The approximation $\exp(y) \approx 1 + y$ becomes increasingly accurate for large $a$ and small $R$. These observations are validated in Fig.~\ref{fig:first_step}, which illustrates how the approximation error diminishes with increasing $a$ and decreases with proximity to the AP (i.e., decreasing $R$).

\begin{figure}[t]
	\centering	\includegraphics[width=0.49\textwidth]{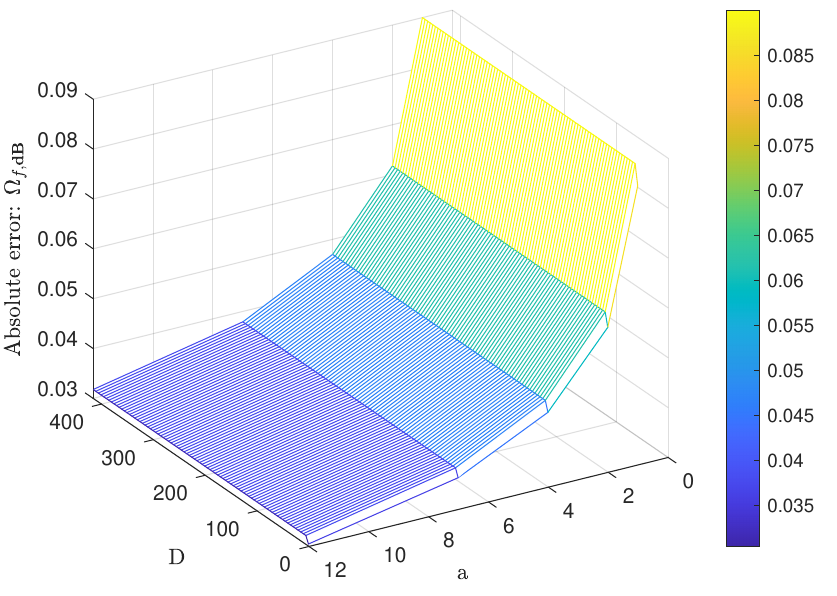}\\
	\caption{Absolute error of $\Omega_{f,\text{dB}}$ under the approximation $e^x \approx 1 + x$.}
	\label{fig:first_step}
\end{figure}

\vspace{0.3em} 
\noindent\textbf{Step 2: Approximating $\boldsymbol{\Omega_{f,\text{dB}}}$ using $\boldsymbol{\frac{1}{1+x}=1-x+R_{1}^{(2)}}$ ($\boldsymbol{x=\frac{B_{2}}{B_{1,k}}R^{-\alpha_{D}}}$) (Eq.(10) to Eq.(11))}

To evaluate the accuracy of the second-order approximation, we define the auxiliary variable
\begin{eqnarray}
&& z(a,R)\triangleq\frac{B_{2}}{B_{1,k}}R^{-\alpha_{D}}=\frac{B_2R^{-\alpha_{D}}}{W+Qa}, \\
&& W\triangleq 1+u_3+u_4+10u_0(\lg{Z_{1,k}}-\frac{1}{\ln10}),\notag\\
&& Q\triangleq u_1+10u_2(\lg{Z_{1,k}}-\frac{1}{\ln10}).\notag
\end{eqnarray}

We first analyze the behavior of the approximation error with respect to the number of feedback symbols $a$. For any fixed $R$, it is straightforward to verify that $z(a, R) > 0$ and its derivative is $z'(a,R)=\frac{Wu_2-Qu_0}{(W+Qa)^2}<0$. Thus, $z(a,R)$ is a monotonically decreasing function of $a$.

The absolute error of the second-order approximation is given by
\begin{equation}
    \delta(z)=\frac{1}{W+Qa}\left[\frac{1}{1+z}-(1-z)\right].
\end{equation}

To analyze the monotonicity, we decompose $\delta(z)$ into two parts. The derivative of the second part $\frac{1}{1+z}-(1-z)$ is
\begin{equation}
     \frac{d[\frac{1}{1+z}-(1-z)]}{dz}= 1-\frac{1}{(z+1)^2}\geq 0,
\end{equation}
which shows that the second part is a monotonically increasing function of $z$. 
Since $z(a, R)$ decreases with increasing $a$, $\frac{1}{1+z}-(1-z)$ is therefore a monotonically decreasing function of $a$. Furthermore, the first multiplicative term $\frac{1}{W+Qa}$ also decreases with $a$. Hence, the total approximation error $\delta(z)$ is a monotonically decreasing function of $a$.

Next, we investigate the dependency on the distance $R$. From the definition, $z(a, R)$ decreases with increasing $R$, which implies that both parts of $\delta(z)$ decrease with $R$. Therefore, the absolute error $\delta(z)$ is also a monotonically decreasing function of $R$.

\textit{Summary:} The second-order approximation $\frac{1}{1 + x} \approx 1 - x$ becomes increasingly accurate as both the number of feedback symbols $a$ and the device distance $R$ increase. This conclusion is supported by simulation results shown in Fig.~\ref{fig:second_step}.

\begin{figure}[t]
	\centering	\includegraphics[width=0.49\textwidth]{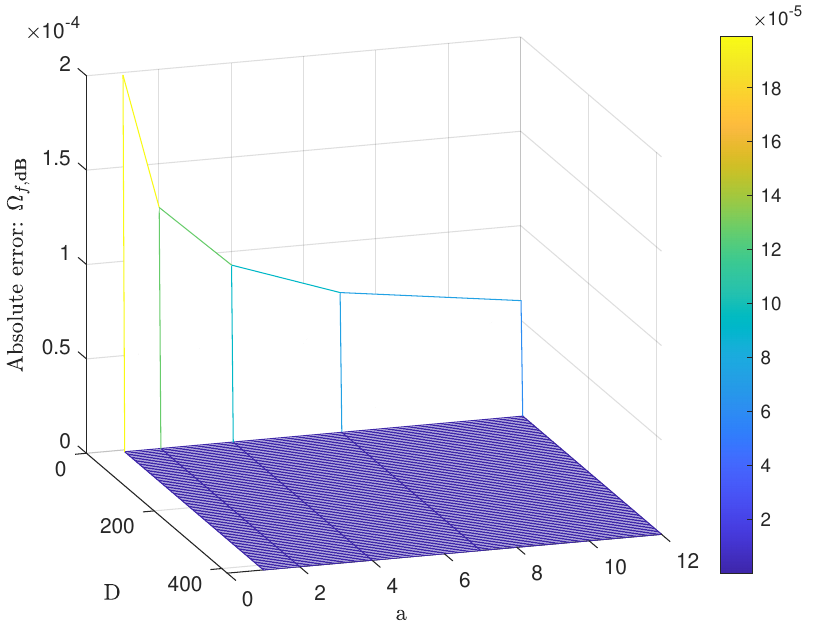}\\
	\caption{Absolute error of $\Omega_{f,\text{dB}}$ under the approximation $\frac{1}{1 + x} \approx 1 - x$.}
	\label{fig:second_step}
\end{figure}

\vspace{0.3em} 
\noindent\textbf{Step 3: Approximating $\boldsymbol{R_1^{(3)}}$ using $\boldsymbol{(1+x)^\alpha\approx{1+\alpha{x}}}$ (Eq.(14) to Eq.(15))}

To assess the accuracy of the third approximation, we define the intermediate variable
\begin{equation}
    v(a,R)\triangleq-\frac{B_{2}R^{-\alpha_D}}{B_{1,k}(1+u_{5}B_{1,k})}.
\end{equation}

We first analyze the behavior of $v(a, R)$ with respect to the number of feedback symbols $a$. The first-order derivative is given by
\begin{eqnarray}
   &&\hspace{-0.8cm} v'(a,R)=\frac{10R^{-\alpha_{D}}}{\ln{10}}\times \\
   &&\hspace{-0.8cm} \frac{u_{2}u_{5}Q^2a^2\!+\!2u_0u_5Q^2a\!-\!u_2W\!-\!u_5u_2W\!+\!u_0Q\!+\!2u_0u_5WQ}{(W+Qa)^2(1+u_{5}W+u_{5}Qa)^2}. \notag
\end{eqnarray}
As can be seen, the demonstrator is strictly positive, and the numerator of $v'(a,R)$ is a quadratic function of $a$. Since $u_{2}u_{5}Q^2<0$ and $-\frac{u_0}{u_2}<0$, the numerator is a decreasing function of $a$ for $a\geq 1$. 
Furthermore, since $v'(1,R)<0$, we have $v'(a,R)<0$, hence $v(a, R)$ is a monotonically decreasing function of $a$.

With the approximation in step 3, the absolute error of $R_1^{(3)}$ is given by
\begin{eqnarray}
&&\hspace{-0.6cm}\delta(v)\!=\!\frac{\ln^2{10}}{200(W\!\!+\!\!Qa)^2}\bigg[1\!+\!u_{5}(W\!+\!Qa)\bigg]^2\bigg[(1\!+\!v)^2\!-\!(1\!+\!2v)\bigg] \notag\\
&&\hspace{-0.6cm}=\frac{\ln^2{10}}{200}\bigg(\frac{1}{W+Qa}+u_{5}\bigg)^2\bigg[(1+v)^2-(1+2v)\bigg].
\end{eqnarray}

We now examine the monotonicity of $\delta(v)$. Since
\begin{equation}
     \frac{d[(1+v)^2-(1+2v)]}{dv}= 2v,
\end{equation}
the second part $(1+v)^2-(1+2v)$ is monotonically increasing in $v$.
Given that $v(a, R)$ is monotonically decreasing in $a$, the second part of $\delta(v)$ is a monotonically decreasing function of $a$. The first part $\frac{\ln^2{10}}{200}(\frac{1}{W+Qa}+u_{5})^2$ also decreases with $a$, hence the overall error $\delta(v)$ is a monotonically decreasing function of $a$.

Regarding distance $R$, since $v(a, R)$ decreases monotonically with $R$, it follows that the error term $\delta(v)$ also decreases with increasing $R$.

\textit{Summary:} The third approximation $(1 + x)^\alpha \approx 1 + \alpha x$ becomes more accurate as both $a$ and $R$ increase. This is confirmed by the simulation results in Fig.~\ref{fig:third_step}, which show that the error diminishes in these regions.

\begin{figure}[t]
	\centering	\includegraphics[width=0.49\textwidth]{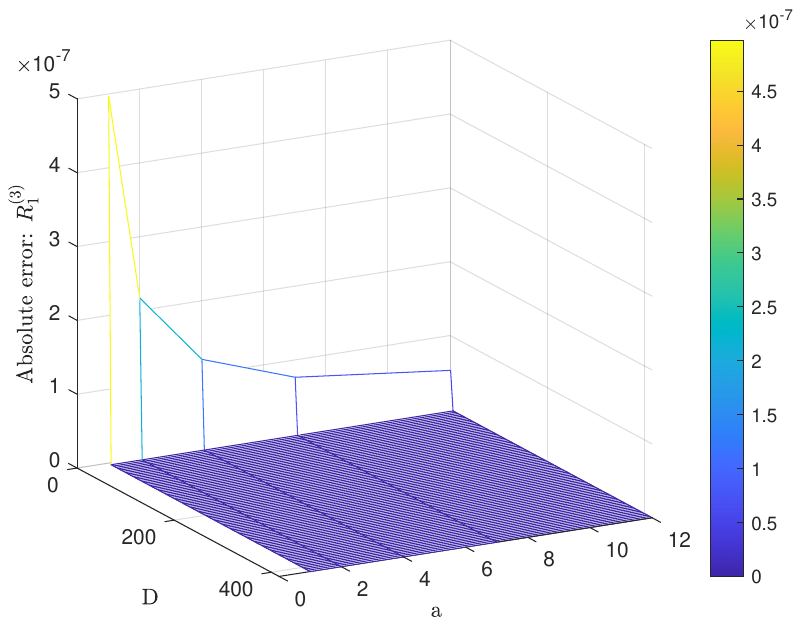}\\
	\caption{Absolute error of $R_1^{(3)}$ under the approximation $(1 + x)^\alpha \approx 1 + \alpha x$.}
	\label{fig:third_step}
\end{figure}

\vspace{0.3em} 
\noindent\textbf{Conclusion of the error analysis}

Overall, the accuracy of the analytical expressions derived in this paper depends on specific ``working points'' of the number of feedback channel uses $a$ and the device-to-AP distance $R$. Across all three approximation steps, we observe that:
\begin{itemize}
    \item Approximation errors consistently decrease with increasing $a$. This holds for each step and leads to a monotonic reduction in the overall approximation error of the main metric $M_f$, the average number of connectable APs.
    \item Regarding distance $R$, the first-step approximation (based on $e^x\approx 1+x$) dominates the total error. While the second and third steps yield relatively minor and tightly bounded errors, the first step introduces a noticeable error that grows with increasing $R$. This implies that the total approximation error of  $M_f$ increases with distance, primarily due to the first step.
\end{itemize}

To quantitatively validate these observations, we conduct a numerical evaluation of the absolute error of $M_f$ across a three-dimensional parameter space: uplink power $P_U$, feedback size $a$, and IoT-to-AP distance $D$. The results are presented in Fig.~\ref{fig:error}.

From Fig.~\ref{fig:error}, we conclude the following:
\begin{itemize}
    \item The approximation is most accurate in regimes with large $a$ and small $D$, which aligns with the analytical trends derived in our step-by-step analysis.
    \item Regarding uplink power $P_U$, a smaller value of $P_U$ leads to a larger $Z_2=\frac{\sigma^2_u}{P_U G_t G_r C_U}$, which in turn amplifies the error in the function $g\left(R,\frac{x_k}{\mu_D}\right)$ in (12), and thus increases the total error in $M_f$. However, this dependence is more moderate compared to the effects of $a$ and $D$.
\end{itemize}

\begin{figure}[t]
	\centering	\includegraphics[width=0.49\textwidth]{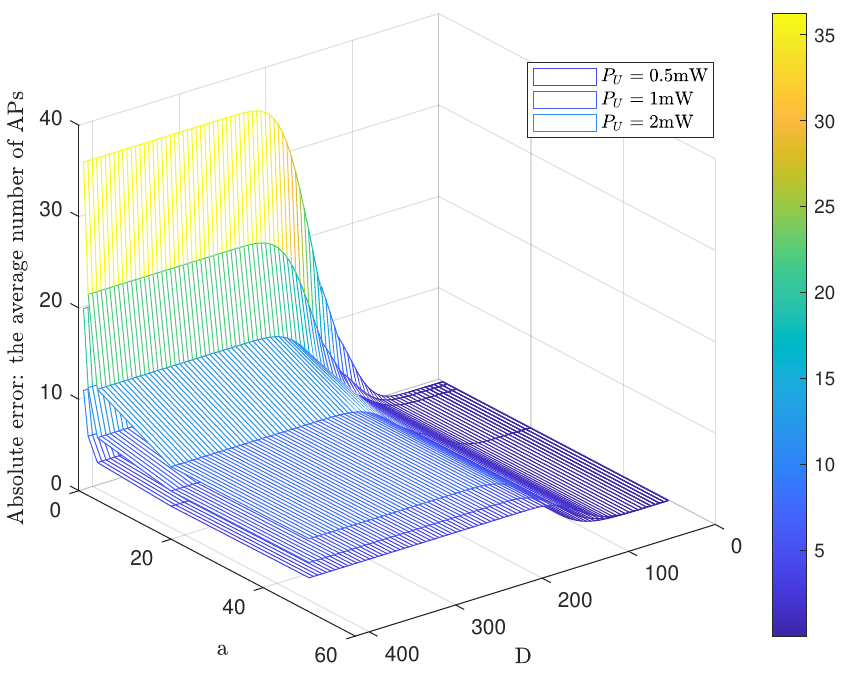}\\
	\caption{Absolute error of the average number of connectable APs as a function of $a$ and $D$.}
	\label{fig:error}
\end{figure}

\begin{figure}[t]
	\centering	\includegraphics[width=0.49\textwidth]{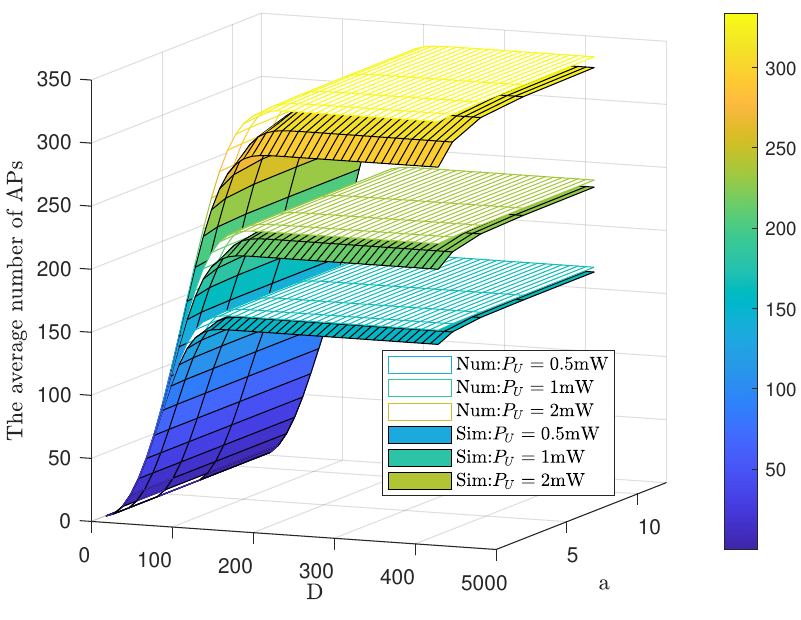}\\
	\caption{The average number of APs. The analytical predictions remain tightly aligned with empirical measurements across various system configurations.}
	\label{fig:accuracy}
\end{figure}

While our analytical expressions do involve several approximations, they are accurate within a wide and practically relevant operating region, particularly when the number of feedback subcarriers $a$ is moderate to large (e.g., $a\geq 2$) and the IoT device is not at extreme distances from the AP (e.g., $D\leq 250$ meters). These conditions are typical in real-world IoT deployments with short packet transmissions and limited transmit power. Moreover, our simulation results in Fig.~\ref{fig:accuracy} confirm that the analytical predictions remain tightly aligned with empirical measurements across various system settings.

\bibliographystyle{IEEEtran}
\bibliography{reference.bib}
\end{document}